\documentclass[runningheads]{llncs}

\usepackage{array}
\usepackage{xspace}
\usepackage[utf8]{inputenc}
\usepackage{graphicx}
\usepackage{amsmath}
\usepackage{float}
\usepackage{amsfonts}
\usepackage{algorithm}      
\usepackage{algpseudocode}  
\usepackage[breaklinks,bookmarks,bookmarksdepth=2]{hyperref}
\usepackage{enumerate}
\usepackage{paralist}
\usepackage{enumitem}
\usepackage[noadjust]{cite}
\usepackage{todonotes}

\newcolumntype{R}[1]{>{\raggedleft\let\newline\\\arraybackslash\hspace{0pt}}m{#1}}
\newcommand{\Oh}{\mathcal{O}}
\newcommand{\freq}{f}
\newcommand{\rank}{\mathsf{rank}}
\newcommand{\select}{\mathsf{select}}
\newcommand{\access}{\mathsf{access}}

\newcommand{\bitvector}{V}
\newcommand{\quadruple}{D}
\newcommand{\quadrupledecomp}{\mathcal{D}}
\newcommand{\block}{B}
\newcommand{\tree}{T}
\newcommand{\inputarr}{\mathsf{A}}

\spnewtheorem{observation}[theorem]{Observation}{\bfseries}{\itshape}

\title{Optimal Query Time for Encoding Range Majority}

\author{Pawe{\l} Gawrychowski and Patrick K. Nicholson}
\institute{University of Haifa, Israel and Nokia Bell Labs, Ireland}

\allowdisplaybreaks

\begin{document}

\pagestyle{plain}
\maketitle

\begin{abstract}
We revisit the range $\tau$-majority problem, which asks us to
preprocess an array $\inputarr[1..n]$ for a fixed value of $\tau \in
(0,1/2]$, such that for any query range $[i,j]$ we can return a position
  in $\inputarr$ of each distinct $\tau$-majority element.  A
  $\tau$-majority element is one that has relative frequency at least
  $\tau$ in the range $[i,j]$: i.e., frequency at least $\tau
  (j-i+1)$.  Belazzougui et al. [WADS 2013] presented a data structure
  that can answer such queries in $\Oh(1/\tau)$ time, which is
  optimal, but the space can be as much as $\Theta(n \lg n)$ bits.
  Recently, Navarro and Thankachan [Algorithmica 2016] showed that
  this problem could be solved using an $\Oh(n \lg (1/\tau))$ bit
  encoding, which is optimal in terms of space, but has suboptimal
  query time.  In this paper, we close this gap and present a data
  structure that occupies $\Oh(n \lg (1/\tau))$ bits of space, and has
  $\Oh(1/\tau)$ query time.  We also show that this space bound is
  optimal, even for the much weaker query in which we must decide
  whether the query range contains at least one $\tau$-majority
  element.
\end{abstract}

\section{Introduction}

Misra and Gries~\cite{MG82} generalized a classic $2$-pass algorithm
by Boyer and Moore~\cite{BM91} for finding majorities in lists of
elements.  Formally, a $\tau$-majority of a list of length $n$ (or
$\tau$-heavy-hitter) is an element that appears with frequency at least
$\tau \cdot n$.  More recent variants and improvements~\cite{DLM02,KSP03} to
the Misra-Gries algorithm have become standard tools in a wide variety
of applications involving streaming analytics, such as IP traffic
monitoring, data mining, etc.

In this paper we consider the data structure variant of the problem.
Suppose we are given an array $\inputarr$ of $n$ elements. The goal is
to preprocess the array into a data structure that supports
\emph{range $\tau$-majority queries}: given an arbitrary subarray
$\inputarr[i..j]$, return all distinct elements that are
$\tau$-majorities in $\inputarr[i..j]$.  As an example application, we
may wish to construct such a data structure on network traffic logs,
to perform an analysis of how the set of frequent users change over
different timescales.

In the last few years, this problem has received a lot of
attention~\cite{KN08,DHMNS13,GHMN11,BGN13}, finally leading to a
recent result of Belazzougui et al.~\cite{BGN13,BGMNN16}: these
queries can be supported in $\Oh(1/\tau)$ time, using
$(1+\varepsilon)nH_0 + o(n)$ bits of space, where $H_0$ is the zero-th
order empirical entropy of the array $\inputarr$, and $\varepsilon$ is
an arbitrary positive constant.\footnote{Note that, for this and all
  forthcoming results discussed, we assume the word-RAM model of
  computation with word-size $w = \Omega(\lg n)$ bits; we use $\lg x$
  to denote $\log_2 x$.  We also note that Belazzougui et
  al.~\cite{BGMNN16} also considered a slightly more difficult problem
  in which $\tau$ can be specified at query time, rather than fixed
  once-and-for-all before constructing the data structure.}  Since,
for an arbitrary $\tau$-majority query, there can be $\lfloor 1/\tau
\rfloor$ answers, there is not much hope for significantly improving
the query time of $\Oh(1/\tau)$, except perhaps to make the time bound
output-sensitive on the number of results
returned~\cite[Sec.7]{BGMNN16}.

On the other hand, much more can be said about the space bound.
Note that, in general, if $\inputarr$ contains elements drawn from the
alphabet $[1,\sigma]$, then we can represent it using $n \lceil
\lg \sigma \rceil$ bits.  If $f_i$ is the frequency of element $i \in
    [1,\sigma]$, then we have $nH_0 = n \sum_i
    \left((f_i/n)\lg(n/f_i)\right) \le n \lceil \lg \sigma
    \rceil$.\footnote{We follow the convention that $(f_i/n)\lg(n/f_i)
      = 0$ if $f_i = 0$.} Since the bound of Belazzougui et
    al.~\cite{BGMNN16} depends on the entropy of the elements in
    $\inputarr$, it can therefore can be $\Theta(n \lg n)$ bits, if
    $\sigma = \Omega(n^c)$ for any constant $c \le 1$, and the
    distribution is close to uniform.  However, quite recently,
    Navarro and Thankachan~\cite{NT16} showed that this space bound
    can be improved significantly in the \emph{encoding model}.

In the encoding model, given array $\inputarr$ as input, we are
allowed to construct an encoding that supports a specific query
operation on $\inputarr$.  After constructing the encoding, the array
$\inputarr$ is deleted, and queries must be supported by accessing
only the encoding.  For many query operations, we can achieve space
bounds that are much smaller than the space required to store
$\inputarr$.  One issue is that for range $\tau$-majority queries, if
we return the actual element which is a $\tau$-majority, then we must
store at least as many bits as are required to represent $\inputarr$.
This follows since an encoding supporting such queries can be used to
return the contents of the array $\inputarr$ by querying the range
$\inputarr[i..i]$ for each $1 \in [1,n]$.

Navarro and Thankachan~\cite{NT16} therefore considered a different
query, in which, for each $\tau$-majority $a$ in the query range
$\inputarr[i..j]$, we instead return an arbitrary position $\ell$ in
$\inputarr$ such that $\inputarr[\ell] = a$ and $i \le \ell \le j$.
In the remainder of the paper, we use \emph{range $\tau$-majority
  position query} to refer to this positional variant of the query
operation.  Navarro and Thankachan~\cite{NT16} showed two main
results:

\begin{theorem}[\cite{NT16}, Theorems 1 and 2]
\begin{enumerate}
\item For any $\tau \in (0,1)$, there is an encoding that occupies
  $\Oh(n \lceil \lg(1/\tau) \rceil)$ bits of space that supports range
  $\tau$-majority position queries in:
  \begin{enumerate}
  \item $\Oh((1/\tau)\lg n)$ time if $1/\tau =
    o(\textup{polylog}(n))$.
  \item $\Oh(1/\tau)$ time if $1/\tau = \Theta(\textup{polylog}(n))$.
  \item $\Oh(1/\tau \lg\lg_w(1/\tau))$ time if $1/\tau =
    \omega(\textup{polylog}(n))$.
  \end{enumerate}
\item Any encoding that can support \emph{range $\tau$-majority
  counting queries} (i.e., return the total the number of
  $\tau$-majorities) in an arbitrary query range $\inputarr[i..j]$
  occupies space (in bits) at least $\frac{n}{4}\left(\lg\left(\frac{1}{2\tau} - 1\right) - \lg e\right) = \Omega(n \lg(1/\tau))$.
\end{enumerate}
\end{theorem}

\noindent Thus, their lower bound implies that their space bound,
which depends only on $n$ and $\tau$ rather than elements in the input
array $\inputarr$, is optimal.  However, there is gap between the
query time of their encoding and the data structure of Belazzougui et
al.~\cite{BGMNN16} for the case where $1/\tau$ is not
$\Theta(\textup{polylog}(n))$.  Crucially, this does not yield optimal
time in the important case where $1/\tau$ is a constant.  In this
paper, we close this time gap, and prove the following theorem:

\begin{theorem}
\label{thm:upper-bound}
For any $\tau \in (0,1/2]$, there is an encoding that occupies $\Oh(n
  \lg (1/\tau))$ bits of space that can answer range $\tau$-majority
  position queries in $\Oh(1/\tau)$ time.\footnote{Note that for $\tau
    \in (1/2,1)$ we can use the $O(n)$ bit structure built for
    $1/2$-majorities to answer all queries.}
\end{theorem}

Of course one could ask if $\Oh(1/\tau)$ is the right bound for the
query time at all. In the output-sensitive variant of the problem the
query time should depend on the number of results returned, which
might be up to $\Oh(1/\tau)$ but possibly smaller. However, we note
that a straightforward reduction from the set intersection conjecture
indicates that a significantly smaller query time cannot be guaranteed
even if the size of the output is 0 or 1: see
Appendix~\ref{setintersection}.

In terms of techniques, our approach uses the level-based
decomposition of Durocher et al.~\cite{DHMNS13}, but with three
significant improvements. We define two new methods for pruning their
data structure to reduce space, and one method to speed up queries.
The first pruning method is a top-down approach that avoids
replicating data structures at more than one level and is analysed
using a charging argument.  The second pruning method is bottom-up,
operating on small ranges of the input array, that we call
micro-arrays, and applies one of two strategies, depending on the
parameter $\tau$.  One of these strategies involves bootstrapping an
optimal space (but suboptimal query time) encoding by combining it
with pre-computed lookup tables in order to speed up queries on the
micro-arrays.  The other strategy stores (a rank reduced) copy of the
micro-array and solves queries in a brute-force manner. Finally, the
last improvement uses wavelet trees~\cite{GGV03} in a non-trivial way
in order to build a fast ranking data structures to improve query time
for the case when $1/\tau = \omega(\textup{polylog}(n))$.

\paragraph{Implications.}
Since the encoding yields the positions of each distinct
$\tau$-majority element in the query range, we can use our optimal
encoding as an alternative to the non-encoding data structure of
Belazzougui et al.  This is done by first compressing the original
array $\inputarr$ using any compressor that supports access in
$\Oh(1)$ to the underlying elements.

\begin{theorem}
Let $\mathcal{S}(n)$ be the space required to store the input array in
a compressed form such that each position can be accessed in $\Oh(1)$
time.  Then there is a data structure that occupies $\mathcal{S}(n) +
\Oh(n \lg(1/\tau))$ bits of space, and can return the range
$\tau$-majorities for an arbitrary range $[i,j]$ in $\Oh(1/\tau)$
time.
\end{theorem}

For example, using results for higher order entropy compression with
$\Oh(1)$ access time~\cite{GN06,FV07} yields the following:

\begin{corollary}
Let $\inputarr[1..n]$ be an array with elements drawn from
$[1,\sigma]$. There is a data structure that occupies $nH_k + o(n \lg
\sigma) + \Oh(n \lg(1/\tau))$ bits of space\footnote{$H_k$ denotes the
  $k$-th order empirical entropy of the sequence of elements in
  $\inputarr$: a lower bound on the space achievable by any compressor
  that considers the previous $k$ elements in the sequence.  For all
  $k \ge 1$ we have $n H_k \le nH_{k-1}$.}, and can support arbitrary
range $\tau$-majority queries in time $\Oh(1/\tau)$, for any $k =
o(\log_\sigma n)$.
\end{corollary}

\paragraph{Lower Bound.} 
Recall the lower bound of $\Omega(n \lg (1/\tau))$ bits holds for any
encoding supporting range $\tau$-majority counting queries.  We
consider an easier problem that we call \emph{range $\tau$-majority
  decision queries}.  The query asks ``Is there at least one element
in the query range $\inputarr[i..j]$ which is a $\tau$-majority?''.
Since the previous lower bound does not rule out a better encoding for
these decision queries, it is natural to ask whether a better encoding
exists.  We prove the following:

\begin{theorem}
\label{thm:lower-bound}
Any data structure that can be used to represent an array
$\inputarr[1..n]$ and support $1/k$-majority decision queries, for any
integer $k \ge 2$, on any arbitrary query range $[i,j]$, requires $n
\lg \frac{k}{e}- \Theta(k^4\lg k)$ bits of space.
\end{theorem}

\noindent Thus, we answer this question in the negative by showing a lower bound
of $\Omega(n \lg(1/\tau))$ bits for any encoding that supports these
queries, which proves our structure is space-optimal for even these
restricted types of queries.  Moreover, we note that our lower bound
has an improved constant factor compared to the previous lower bound.

\paragraph{Related Work.}
Finally, we remark that the area of range queries on arrays is quite
vast, and there are many interesting related types of queries that
have been studied in the both the non-encoding and encoding models; we
refer the reader to surveys on the topics~\cite{S13,R15}.  The most
closely related problem to the range $\tau$-majority problem is the
\emph{range mode problem}~\cite{CDLMW14}: given a query range $[i,j]$
return the most frequently appearing element in the range.  In
contrast with range $\tau$-majority, this type of query is
significantly less efficient, with the best $\Theta(n \lg n )$ bit
data structures having $\Oh(\sqrt{n/\lg n})$ query time.

\section{Preliminaries}


\begin{lemma}[\cite{P08}]
\label{lem:rank-select}
Let $\bitvector$ be a bit vector of length $n$ bits in which $m$ of
the bits are set to one. There is a data structure for representing
$\bitvector$ that uses $m\lg(n/m) + \Oh(n/\lg^c(n))$ bits for any
constant $c \ge 1$ such that the following queries can be answered in
$\Oh(1)$ time:
\begin{itemize}
\item $\access(\bitvector,i)$ returns bit $\bitvector[i]$.
\item $\rank(\bitvector,i)$: returns the number of ones in the prefix
  $\bitvector[1..i]$.
\item $\select(\bitvector,j)$: returns the index of the
  $j$-th one in $\bitvector$, if it exists, and $-1$
  otherwise. In other words, the inverse of the rank operation: if
  $\select(\bitvector,j) = i$, then $\rank(\bitvector,i) = j$.
\end{itemize}
\end{lemma}

\noindent Since our proof makes heavy use of this lemma, we
distinguish the $m\lg(n/m)$ term in the space bound by calling it the
\emph{leading term}, and the other term the \emph{redundancy}.  If we
do not need the full power of rank, then we can use the following
lemma to reduce the redundancy:
\begin{lemma}[\cite{RRR07}]
\label{lem:sel-only}
If only the constant time select and access operations are required,
then we can represent $\bitvector$ using $m \lg(n/m) + o(m) + \Oh(\lg
\lg n)$ bits.
\end{lemma}

A useful fact about applying these previous Lemmas to bit vectors is
that concatenation is often helpful: if we apply either Lemma to two
bit vectors separately, both of length $n$ containing at least $m$
bits, then the sum of the leading terms is no more than $2m \lg
(n/m)$.  If we concatenate the bit vectors before applying the lemma,
the upper bound on the leading term is the same.

\section{Upper Bound}

\subsection{Quadruple Decomposition}

The upper bound makes use of the \emph{quadruple decomposition} of
Durocher et al.~\cite{DHMNS13}.  For ease of description, we assume
that $n$ is a power of $2$, but note that decomposition works in
general.  First, at a conceptual level we build a balanced binary tree
over the array $\inputarr[1..n]$.  Each leaf represents an element
$\inputarr[i]$.  On the $k$-th level of the tree $\tree(k)$, counting
from the leaves at level $0$, the nodes represent a partition of
$\inputarr[1..n]$ into $n_k = n / 2^k$ contiguous \emph{blocks} of
length $2^k$.  Second, consider all levels containing at least four
blocks.  At each such level, consider the blocks $\block_1, \ldots,
\block_{n_k}$.  We create a list of \emph{quadruples} (i.e., groups of
four consecutive blocks) at each such level:
\begin{align*}
\quadrupledecomp(k) = [(\block_1, \block_2, \block_3, \block_4), & (\block_3, \block_4, \block_5, \block_6), \ldots, (\block_{n_k-1}, \block_{n_k}, \block_1, \block_2)].
\end{align*}
\noindent
Thus, each index in $\inputarr$ is contained in exactly two quadruples
at each level, and there is one quadruple that wraps-around to handle
corner cases.  The quadruples are staggered at an offset of two blocks
from each other.  Moreover, given a quadruple $\quadruple =
(\block_{2\ell+1},\block_{2\ell+2},\block_{2\ell+3},\block_{2\ell+4})$,
the two middle blocks $\block_{2\ell+2}$ and $\block_{2\ell+3}$ are
not siblings in the binary tree $\tree$.  We call the range spanned by
these two middle blocks the \emph{middle part} of $\quadruple$.

As observed by Durocher et al.~\cite{DHMNS13}, for every query range
$[i,j]$ there exists a unique level $k$ in the tree such that $[i,j]$
contains at least one and at most two consecutive blocks in
$\tree(k)$, and, if $[i,j]$ contains two blocks, then the nodes
representing these blocks are not siblings in the tree $\tree$.  Thus,
based on our quadruple decomposition, for every query range $[i,j]$ we
can \emph{associate} it with exactly one quadruple $\quadruple =
(\block_{2\ell+1},\block_{2\ell+2},\block_{2\ell+3},\block_{2\ell+4})$ such
that
\[\left(\left( \block_{2\ell+2} \subseteq [i,j]\right) \lor \left( \block_{2\ell+3} \subseteq [i,j]\right)\right) \land \left( [i,j] \subset \block_{2\ell+1} \cup \block_{2\ell+2} \cup \block_{2\ell+3} \cup \block_{2\ell+4}.\right)\]
Moreover, Durocher et al.~\cite{DHMNS13} proved the following lemma:

\begin{lemma}[\cite{DHMNS13}]\label{lem:find-level}
For each query range $[i,j]$, in $\Oh(1)$ time we can compute the level
$k$, as well as the offset of the quadruple associated with $[i,j]$ in
the list $\quadrupledecomp(k)$, using $o(n)$ bits of space.  
\end{lemma}

Furthermore, if we consider any arbitrary query range $[i,j]$ that is
associated with a quadruple $\quadruple$, there are at most $4/\tau$
elements in the range represented by $\quadruple$ that could be
$\tau$-majorities for the query range.  Following Durocher et al., we
refer to these elements as \emph{candidates} for the quadruple
$\quadruple$.

For each quadruple, we compute and store all of its candidates, so
that, by Lemma~\ref{lem:find-level}, in $\Oh(1)$ time we can obtain
$\Oh(1/\tau)$ candidates. It remains to show how to verify that a
candidate is in fact a $\tau$-majority in $\inputarr[i..j]$.  At this point,
our approach deviates from Durocher et al.~\cite{DHMNS13}, who
make use of a wavelet tree for verification, and end up with a space
bound of $\Oh(n \lg n \lg(1/\tau))$ bits.

Consider such a candidate $y$ for quadruple
$\quadruple=(\block_{2\ell+1},\ldots,\block_{2\ell+4})$.  Our goal is
to count the number of occurrences of $y$ in the query range $[i,j]$.
To do this we store a bit vector $\bitvector(\quadruple,y)$, that
represents the (slightly extended) range $\block_{2\ell} \cup \ldots
\cup \block_{2\ell+5}$ and marks all occurrences of $y$ in this range
with a one bit.  By counting the number of ones in the range
corresponding to $[i,j]$ in $\bitvector(\quadruple,y)$, we can
determine if the number of occurrences exceeds the threshold
$\tau(j-i+1)$.  If the threshold is not exceeded, then we can return
the first one bit in the range, as that position in $\inputarr$
contains element $y$. Note that we have extended the range of the bit
vector beyond the range covered by $\quadruple$ by one extra block to
the left and right.  We call this extended range the \emph{extent} of
$\quadruple$, and we make the following observation (clearly visible
in Appendix~\ref{sec:figures}).

\begin{observation}
\label{obs:extent}
Let $E(\quadruple)$ be the extent of quadruple $\quadruple$ at level
$k$.  Then for all quadruples $\quadruple'$ at level $k' < k$ such
that the range of $\quadruple'$ has non-empty intersection with the
range of $\quadruple$, we have that $\quadruple' \subset
E(\quadruple)$.
\end{observation}

We now briefly analyze the total space of this method, under the
assumption that we can store a bit vector of length $n$ with $m$ one
bits using $\Oh(m \lg(n/m))$ bits.  This crude analysis is merely to
illustrate that additional tricks are needed to achieve optimal
space. The quadruple decomposition consists of $\lg n$ levels. On each
level, we store a number of bit vectors. For each quadruple we have up
to $\Oh(1/\tau)$ candidates $\mathcal{Y}$.  Thus, if $f_y$ represents
the frequency of candidate $y$ in extent of quadruple $\quadruple$,
then the space bound, for each quadruple at level $k$, is
$\sum_{y\in\mathcal{Y}}\Oh(f_c\lg(2^k/f_c)) $, which, by the concave
version of Jensen's inequality, is bounded by $\Oh((2^k)\lg(1/\tau))$.
So each level uses $\Oh(n \lg(1/\tau))$ bits, for a total of $\Oh(n
\lg n \lg(1/\tau))$ bits over all levels.

\subsection{\label{sec:low-space}Optimal Space with Suboptimal Query Time}

To achieve space $\Oh(n \lg (1/\tau))$ bits, the intuition is that we
should avoid duplicating the same bit vectors between levels.
It is easy to imagine a case where element $y$ is a candidate at every
level and in every quadruple of the decomposition, which results in
many duplicated bit vectors.  To avoid this duplication problem, we
propose a top-down algorithm for pruning the bit vectors.  Initially,
all indices in $\inputarr$ are \emph{active} at the beginning.  Our
goal is to charge at most $\Oh(\lg(1/\tau))$ bits to each active index
in $\inputarr$, which achieves the desired space bound.

Let $k$ be the current level of the quadruple decomposition, as
proceed top-down. We maintain the invariant that for any element $y$
in a block $\block_i$, either all indices storing occurrences of $y$
are active in $\block_i$ (in which case we say $y$ is active in
$\block_i$), or none are (in which case we say $y$ is inactive in
$\block_i$).  Consider a candidate $y$ associated with quadruple
$\quadruple = (\block_{2\ell +1}, \block_{2\ell + 2},
  \block_{2\ell+3}, \block_{2\ell+4})$. Then:
  
\begin{enumerate}
\item If $y$ is active in blocks $\block_{2\ell +1}, \ldots,
  \block_{2\ell +4}$, then we store the bit vector
  $\bitvector(\quadruple,y)$, and (conceptually) mark all occurrences
  of $y$ inactive in these blocks \emph{after} we finish processing
  level $k$.  This makes $y$ inactive in all blocks contained in
  $\quadruple$ at lower levels.  Since a block $\block_i$ is contained
  in two quadruples at level $k$, a position storing $y$ in $\block_i$
  may be made inactive for two reasons: this is why we mark positions
  inactive after processing all quadruples at level $k$.
\item If $y$ is inactive in some block $\block_i \subset \quadruple$,
  then it is the case that we have computed and stored the bit vector
  $\bitvector(\quadruple',y)$ for some quadruple $\quadruple'$ at
  level $k'>k$, such that $\quadruple \cap \quadruple' \neq 0$.
  Therefore, Observation~\ref{obs:extent} implies that $\quadruple$ is
  contained in the extent of $\quadruple'$, and thus the bit vector
  associated with $\quadruple'$ can be used to answer queries for
  $\quadruple$.  For $\quadruple$ we need not to store
  $\bitvector(\quadruple,y)$, though for now we do not address how to
  efficiently answer these queries.
\end{enumerate}

Next we analyse the total cost of the bit vectors that we stored
during the top-down construction.  The high level idea is that we can
charge the cost of bit vector $\bitvector(\quadruple,y)$ to the
indices in $\quadruple$ that store occurrences of $y$.  Call these the
indices the \emph{sponsors} of $\bitvector(\quadruple,y)$.  Since $y$
is a $\tau$-majority, it occurs at least $\Oh(\tau \cdot 2^k)$ times in
$\quadruple$, which has length $\Oh(2^k)$.  Thus, we can expect to
charge $\Oh(\lg(1/\tau))$ bits to each sponsor: the expected gap
between one bits is $\Oh(1/\tau)$ and therefore can be recorded using
$\Oh(\lg(1/\tau))$ bits.  There are some minor technicalities that
must be addressed, but this basic idea leads to the following
intermediate result, in which we don't concern ourselves with the
query time:

\begin{lemma}
\label{lem:encoding-structure}
There is an encoding of size $\Oh(n \lg(1/\tau))$ bits such that the
answer to all range $\tau$-majority position queries can be recovered.
\end{lemma}

\begin{proof}
Consider candidate $y$ and its occurrences in extent $E(\quadruple)$
of quadruple $\quadruple$ at level $k$, for which we stored the bit
vector $\bitvector(\quadruple,y)$.  Suppose there are $f_y$
occurrences of $y$ in $E(\quadruple)$.  If at least one third of the
occurrences of $y$ are contained in $\quadruple$, then we charge the
cost of the bit vector to the (at least) $f_y/3$ sponsor indices in
$\quadruple$.  Otherwise, this implies one of the two blocks, call it
$\block_i$ such that $\block_i \subset E(\quadruple)$ but $\block_i
\not\subset \quadruple$ contains at least $f_y/3$ occurrences of $y$.
Therefore, $y$ must also be an active candidate for the unique
quadruple $\quadruple'$ that has non-empty intersection with both
$\block_i$ and $\quadruple$: this follows since $y$ occurs more times
in $\quadruple'$ than in $\quadruple$, and $y$ is a candidate for
$\quadruple$.  In this case we charge the cost of the bit vector to
the sponsor indices in \emph{neighbouring quadruple} $\quadruple'$.

Suppose we store the bit vectors using Lemma~\ref{lem:sel-only}: for
now ignore the $\Oh(\lg \lg n)$ term in the space bound as we
deal with it in the next paragraph.  Using Lemma~\ref{lem:sel-only},
the cost of the bit vector $\bitvector(\quadruple,y)$ associated with
$\quadruple$ is at most $\Oh(f_y\lg(1/\tau))$, since $y$ is a
$(\tau/4)$-majority in $\quadruple$.  Thus, $\Oh(f_y)$ sponsors in
$\quadruple$ pay for at most three bit vectors:
$\bitvector(\quadruple,y)$ and possibly the two other bit vectors that
cost $\Oh(f_y \lg (1/\tau))$ bits, charged by neighbouring quadruples.
Since this charge can only occur at one level in the decomposition
(the index becomes inactive at lower levels after the first charge
occurs), each sponsor is charged $\Oh(\lg(1/\tau))$, making the total
amount charged $\Oh(n \lg(1/\tau))$ bits overall.

To make answering queries actually possible, we make use of the same
technique used by Durocher et al.~\cite{DHMNS13}, which is to
concatenate the bit vectors at level $k$.  The candidates have some
implicit ordering in each quadruple, $[1, ..., \Oh(1/\tau)]$: the
ordering can in fact be arbitrary.  For each level $k$, we concatenate
the bit vectors associated with quadruple according to this implicit
ordering of the candidates.  Thus, since there are $\Oh(\lg n)$ bit
vectors (one per level), the $\Oh(\lg \lg n)$ term for
Lemma~\ref{lem:sel-only} contributes $\Oh(\lg n \lg \lg n)$ to the
overall space bound.

Given a query $[i,j]$, Lemma~\ref{lem:find-level} allows us to compute
the level $k$ and offset $\ell$ of the quadruple associated with
$[i,j]$.  Our goal is to remap $[i,j]$ to the relevant query range in
the concatenated bit vector at level $k$.  Since all bit vectors
$\bitvector(\quadruple, y)$ at level $k$ have the same length, we only
need to know how many bit vectors are stored for quadruples $1, ...,
\ell -1$: call this quantity $X$.  Thus, at level $k$ we construct and
store a bit vector $L_k$ of length $\Oh(n_k/\tau)$ in which we store
the number of bit vectors associated with the quadruples in unary.
So, if the first three quadruples have $2,6,4$ candidates
(respectively), we store $L_k = 1001000000100001\ldots$.  Overall, the
space for $L_k$ is $\Oh(n_k \lg(1/\tau))$, or $\Oh(n \lg(1/\tau))$
overall, if we represent each $L_k$ using
Lemma~\ref{lem:sel-only}.

Given an offset $\ell$, we can perform $\select(L_k,\ell) - \ell$ to
get $X$.  Once we have $X$, we can use the fact that all extents have
fixed length at a level in order to remap the query $[i,j]$ to the
appropriate range $[i',j']$ in the concatenated bit vector for each
candidate.  We can then use binary search and the select operation to
count the number of $1$ bits corresponding to each candidate in the
remapped range $[i',j']$ in $\Oh(\lg n)$ time per candidate.  Since
some of the candidates for the $\quadruple$ associated with $[i,j]$
may have been inactive, we also must compute the frequency of each
candidate in quadruples at higher levels that contain $\quadruple$.
Since there are $\Oh(\lg n)$ levels, $\Oh(1)$ quadruples that overlap
$\quadruple$ per level, and $\Oh(1/\tau)$ candidates per quadruple, we
can answer range $\tau$-majority position queries in $\Oh(\lg^2
n/\tau)$ time.  Note that we have to be careful to remove possible
duplicate candidates (at each level the quadruples that overlap
$\quadruple$ may share candidates).
\end{proof}

\subsection{\label{sec:fast-query}Optimal Space with Optimal Query Time}

In Lemma~\ref{lem:encoding-structure} there are two issues that make
querying inefficient: 1) we have to search for inactive candidates in
$\Oh(\lg n)$ levels; and 2) we used Lemma~\ref{lem:sel-only} which
does not support $\Oh(1)$ time rank queries.  The solutions to both of
these issues are straightforward.  For the first issue, we store
pointers to the appropriate bit vector at higher levels, allowing us
to access them in $\Oh(1)$ time.  For the second issue we can use
Lemma~\ref{lem:rank-select} to support rank in $\Oh(1)$ time.
However, both of these solutions raise their own technical issues
that we must resolve in this section.

\paragraph{Pointers to higher levels.} 
Consider a quadruple $\quadruple$ at level $k$ for which candidate $y$
is inactive in some block contained in $\quadruple$.  Recall that this
implies the existence of some bit vector $\bitvector(\quadruple',y)$
for some $\quadruple'$ at level $k'> k$ that can be used to count
occurrences of $y$ in $\quadruple$.  In order to access this bit
vector in $\Oh(1)$ time, the only information that we need to store is
the number $k'$ and also the offset of $y$ in the list of candidates
for $\quadruple'$: $\quadruple'$ might have a different ordering on
its candidates than $\quadruple$.  Thus, in this case we store
$\Oh(\lg\lg n + \lg(1/\tau))$ bits per quadruple as we have $\Oh(\lg
n)$ levels and $\Oh(1/\tau)$ candidates per quadruple.  This is a
problem, because there are $\Oh(n)$ quadruples, which means these
pointers can occupy $\Oh(n/\tau (\lg \lg n + \lg(1/\tau)))$ bits overall.

To deal with this problem, we simply reduce the number of quadruples
using a bottom-up pruning technique: all data associated with
quadruples spanning a range of size $Z$ or smaller is deleted.  This
is good as it limits the space for the pointers to at most
$\Oh(n(\lg\lg n + \lg(1/\tau))/(\tau \cdot Z))$ bits, as there are
$\Oh(n/Z)$ quadruples of length greater than $Z$.  However, we need to
come up with an alternative approach for queries associated with these
small quadruples.

The value we select for $Z$, as well as the strategy to handle queries
associated with quadruples of size $Z$ or smaller, depends on the
value of $1/\tau$:

\begin{enumerate}
\item If $1/\tau \ge \sqrt{\lg n}$: then we set $Z = 1/\tau$.  Thus,
  the pointers occupy $\Oh(n(\lg \lg n + \lg(1/\tau)) =
  \Oh(n\lg(1/\tau))$ bits (since $\lg(1/\tau) = \Omega(\lg \lg
  n)$). Consider the maximum level $k$ such that the quadruples are of
  size $Z$ or smaller.  For each quadruple $\quadruple$ in level $k$,
  we construct a new \emph{micro-array} of length $2^k$ by copying the
  range spanned by $\quadruple$ from $\inputarr$. Thus, any query
  $[i,j]$ associated with a quadruple at levels $k$ or lower can be
  reduced to a query on one of these micro-arrays.  Since the
  micro-arrays have length $1/\tau$, we preprocess the elements in the
  array by replacing them by their ranks (i.e., we reduce the elements
  to rank space).  Storing the micro-array therefore requires only
  $\Oh(n_k 2^k \lg(1/\tau)) = \Oh(n \lg(1/\tau))$ bits.  Moreover,
  since we have access to the ranks of the elements directly, we can
  answer any query on the micro-array directly by scanning it in
  $\Oh(1/\tau)$ time.  Thus, in this case, the space for the
  micro-arrays and pointers is $\Oh(n\lg(1/\tau))$.
\item If $1/\tau < \sqrt{\lg n}$: in this branch we use the encoding
  of Lemma~\ref{lem:encoding-structure} that occupies $c\cdot n\lg(1/\tau)$
  bits of space for an array of length $n$, for some constant $c \ge 1$.  We set $Z = \lg n /
  (2c \lg(1/\tau))$, so that the space for the pointers becomes:
\[ \Oh(n(\lg \lg n + \lg(1/\tau))\lg(1/\tau)/(\tau \cdot \lg n)) = \Oh(n (\lg \lg n)^2 / \sqrt{\lg n}) = o(n) .\]
  As in the previous case, we construct the micro-arrays for the
  appropriate quadruples based on the size $Z$.  However, this time we
  encode each micro-array using
  Lemma~\ref{lem:encoding-structure}. This gives us a set of $n_k$
  encodings, taking a total $\Oh(n_k 2^k\lg(1/\tau)) = \Oh(n
  \lg(1/\tau))$ bits.  Moreover, the answer to a query is fully
  determined by the encoding and the endpoints $i,j$.  Since $i$ and
  $j$ are fully contained in the micro-array, their description takes
  $\lg Z$ bits.  Thus, using an auxiliary lookup table of size
  $\Oh(2^{c\cdot Z \lg(1/\tau)}\times 2^{\lg^{2}Z})$ we can preprocess the
  answer for every possible encoding and positions $i,j$ so that a
  query takes $\Oh(1)$ time. Because $1/\tau < \sqrt{\lg n}$ the space for
  this lookup table is:
\[ \Oh(2^{c(\lg n/(2c \lg(1/\tau)))\lg(1/\tau) + \lg^2(\lg n / (2c \lg(1/\tau)))}) = \Oh(2^{\lg n / 2 + (\lg \lg n)^2}) = o(n). \]
\end{enumerate}

\noindent In summary, we can apply level-based pruning to reduce the
space required by the pointers to at most $\Oh(n \lg(1/\tau))$.  Note
that we must be able to quickly access the pointers associated with
each quadruple $\quadruple$.  To do this, we concatenate the pointers
at level $k$, and construct yet another bit vector $L'_k$ having a
similar format as $L_k$.  The bit vector $L'_k$ allows us to easily
determine how many pointers are stored for the quadruples to the left
of $\quadruple$ at the current level, as well as how many are stored
for $\quadruple$.  Thus, these additional bit vectors occupy $\Oh(n
\lg(1/\tau))$ bits of space, and allow accessing an arbitrary pointer
in $\Oh(1)$ time.

\paragraph{Using the faster ranking structure.} 
When we use the faster rank structure of Lemma~\ref{lem:rank-select},
we immediately get that we can verify the frequency of each candidate
in $\Oh(1)$ time, rather than $\Oh(\lg n)$ time.  Recall that the bit
vectors are concatenated at each level.  In the structure of
Lemma~\ref{lem:sel-only}, the redundancy at each level was merely
$\Oh(\lg \lg n)$ bits.  However, with Lemma~\ref{lem:rank-select} we
end up with a redundancy of $\Oh(n /(\tau \lg^c (n)))$ bits per level,
for a total of $\Oh(n \lg n/(\tau \lg^c (n)))$ bits.  So, if $1/\tau =
\Oh(\textup{polylog}(n))$, then we can choose the constant $c$ to be
sufficiently large so that this term is sublinear.  Immediately, this
yields:

\begin{lemma}
If $1/\tau = \Oh(\textup{polylog}(n))$, there is an encoding that
supports range $\tau$-majority position queries in $\Oh(1/\tau)$ time,
and occupies $\Oh(n \lg(1/\tau))$ bits.
\end{lemma}

\noindent When $1/\tau$ is $\omega(\textup{polylog(n)})$, we require a
more sophisticated data structure to achieve $\Oh(n \lg(1/\tau))$ bits
of space.  Basically, we have to replace the data structure of
Lemma~\ref{lem:rank-select} representing the bit vectors with a more
space-efficient batch structure that groups all candidates together.
We present the details in Appendix~\ref{sec:handling-small-tau}.  This
data structure allows us to complete Theorem~\ref{thm:upper-bound}.

\section{Lower Bound}

In this section we prove Theorem~\ref{thm:lower-bound}.  The high
level idea is to show that we recover a sequence of concatenated
permutations of length roughly $1/\tau$ each using the query
operation.  This requires a more refined padding argument than that
presented by Navarro and Thankachan~\cite{NT16}.

Formally, we will describe a \emph{bad string}, defined using
concatenation, in which array $\inputarr[i]$ will store the $i$-th
symbol in the string.  Conceptually, this bad string is constructed by
concatenating some \emph{padding}, denoted $L$, before a sequence of
$m$ permutations over the alphabet $[\alpha_1,\ldots,\alpha_k]$,
denoted $R = \pi_1\cdot \ldots \cdot \pi_m$.  Notationally, we use
$\alpha_i^c$ to denote a concatenation of the symbol $\alpha_i$ $c$
times, and $a \cdot b$ to denote the concatenation of the strings $a$
and $b$.  In the construction we make use of \emph{dummy symbols},
$\beta$, which are defined to be symbols that occur exactly one time
in the bad string.  A sequence of $\ell$ dummy symbols, written
$\beta^{\ell}$, should be taken to mean: a sequence of $\ell$
characters, each of which are distinct from any other symbol in
$\inputarr$.

\paragraph{Padding definition.} Key to defining $L$ is a gadget $G(k,i)$,
that is defined for any integer $k \ge 2$ using concatenation as
follows: $G(k,i) = \alpha_1^{k'}\cdot\alpha_2^{k'}\cdot \ldots \cdot
\alpha_{i-1}^{k'} \cdot \alpha_{i+1}^{k'} \cdot \ldots \cdot
\alpha_{k}^{k'} \cdot (\alpha_i \cdot \beta^{k-2})^{k-1}\cdot \alpha_i
\beta^k$, where $k' = k^2 -k + 2$. An example in which $k=2$ can be
found in Figure~\ref{fig:lower-bound} in Appendix~\ref{sec:figures}.
Suppose we define $\inputarr$ such that $\inputarr[x..y]$ contains
gadget $G(k,i)$.  Let $\freq(x,y,\alpha)$ denote the number of
occurrences of symbol $\alpha$ in range $[x,y]$.  We define the
\emph{density} of symbol $\alpha$ in the query range $[x,y]$ to be
$\delta(x,y,\alpha) = \freq(x,y,\alpha) / (y-x+1)$.  We observe the following:

\begin{enumerate}
\item The length of the gadget $G(k,i)$ is $k(k^2-k+2)$ for all $i \in
  [1,k]$.  This fact will be useful later when we bound the total size
  of the padding $L$.
\item $\delta(x,y,\alpha_j) = 1/k$ for all $j \neq i$.  This follows
  from the previous observation and that, for all $j \neq i$,
  the number of occurrences of $\alpha_j$ in $G(k,i)$ is $k^2-k+2$.
\end{enumerate}

Next, we finish defining our array $\inputarr$ by defining $L$ to be
the concatenation $G(k,k) \cdot G(k,k-1) \cdot \ldots \cdot G(k,1)$.
Thus, our array is obtained by embedding the string $L \cdot R$ into
an array $\inputarr$.  Note that the total length of the array is
$k^2(k^2-k+2) + mk$.  Thus, the padding is of length $\Theta(k^4)$.

\paragraph{Query Procedure.} The following procedure can recover the
position of symbol $\alpha_i$ in $\pi_j$, for any $i \in [1,k]$ and $j
\in [1,m]$.  This procedure uses $\Theta(k)$ $(1/k)$-majority decision
queries: overall, recovering the contents of $R$ uses $\Theta(k^2m)$
queries.

Let $r_{j,1}, ..., r_{j,k}$ denote the indices of $\inputarr$ containing the
symbols in $\pi_j$ from left-to-right.  Moreover, consider the indices
of the $k$ occurrences of symbol $\alpha_i$ in $G(k,i)$, from
left-to-right, and denote these as $\ell_{i,k}, \ldots, \ell_{i,1}$,
respectively (note that the rightmost occurrence is marked with
subscript $1$).  See Figure~\ref{fig:lower-bound} for an
illustration.  Formally, the query procedure will perform a sequence
of queries, stopping if the answer is \texttt{YES}, and continuing if
the answer is \texttt{NO}.  The sequence of queries is
$[\ell_{i,1},r_{j,1}],[\ell_{i,2},r_{j,2}],\ldots,[\ell_{i,k},r_{j,k}]$.

We now claim that if the answer to a query $[\ell_{i,x},r_{j,x}]$ is
\texttt{NO}, then $\inputarr[r_{j,x}] \neq \alpha_i$.  This follows since the
density of symbol $\alpha_i$ in the query range is:
\begin{align*}
\frac{x + (i-1)(k^2 - k + 2) + (j-1)}{k(x + (i-1)(k^2 -k +2) + (j-1)) + 2} < \frac{1}{k}
\end{align*}
\noindent
On the other hand, if the answer is \texttt{YES}, we have that the
symbol $\alpha_i$ must be a $(1/k)$-majority for the following reasons:
\begin{enumerate}
\item No other symbol $\alpha_j$ where $j \neq i$ can be a
  $(1/k)$-majority. To see this, divide the query range into a
  middle-part, consisting of $G(k,i-1) \cdot \ldots \cdot G(k,1) \cdot
  \pi_1 \cdot \ldots \cdot \pi_{k-1}$, as well as a prefix (which is a
  suffix of $G(k,i)$), and a suffix (which is a prefix of
  $\pi_j$). The prefix of the query range contains no occurrence of
  $\alpha_j$ and is at least of length $k+1$.  The suffix contains at
  most one occurrence of $\alpha_j$.  Thus, the density of $\alpha_j$
  is strictly less than $1/k$ in the union of the prefix and suffix,
  exactly $1/k$ in the middle part, and strictly less than $1/k$
  overall.
\item No dummy symbol $\beta$ can be an $(1/k)$-majority,
  since these symbols appear one time only, and all query ranges have
  length strictly larger than $k$.
\item Finally, if $\inputarr[r_{j,x}] = \alpha_i$, then the density
  $\delta(\ell_{i,x},r_{j,x},\alpha_i)$ is: 
  $$\frac{x + (i-1)(k^2 - k + 2) + (j-1) +1}{k(x + (i-1)(k^2 -k +2) + (j-1)) + 2} \ge 1/k,$$
  \noindent 
  since $k \ge 2$.  Since we stop immediately after the first
  \texttt{YES}, the procedure therefore is guaranteed to identify the
  correct position of $\alpha_i$.
\end{enumerate}  

As we stated, the length of the array is $k^2(k^2-k+2) + mk = n$, and
for $n$ large enough the queries allow us to recover
$\frac{n-\Theta(k^{4})}{k} \lg(k!)$ bits of information using
$(1/k)$-majority queries for any integer $k \ge 2$, which is at least
$(n/k - \Theta(k^3))k \lg (k/e) = n\lg (k/e) - \Theta(k^4 \lg k)$
bits. Since there exists a unit fraction $\tau' = 1/\lfloor 1/\tau
\rfloor$ (if $\tau \in (0,1/2]$), there also exists a bad input of
  length $n$ in which $k = 1/\tau'$. Therefore, we have proved
  Theorem~\ref{thm:lower-bound}.

\bibliographystyle{splncs03} 
\bibliography{biblio}

\newpage
\appendix

\section{Hardness of the output-sensitive variant}
\label{setintersection}

P\v{a}tra\c{s}cu and Roditty~\cite{PatrascuR14} state the following folklore set intersection
conjecture:

\begin{conjecture}
Consider a data structure that preprocesses sets $S_{1},\ldots,S_{n}\subseteq [X]$,
and answers queries of the form “does $S_{i}$ intersect $S_{j}$?”. Let $X = \lg^{c} n$
for a large enough constant $c$. If the query takes constant time, the space must be
$\Omega(n^{2})$.
\end{conjecture}

\noindent
They mention that even for queries taking $\Omega(|X|/\lg n)$ time the
conjecture is plausible.

We point out the following a simple connection between the set
intersection conjecture and a structure supporting range
$\tau$-majority decision queries.  Given an instance of the set
intersection problem, we construct $\inputarr$ of length
$(2n+2)X$. First, for every set $S_{i}$, we define a string $B_{i}$ by
writing down elements of $X$ that do not belong to $S_{i}$ and then
the elements that do belong to $S_{i}$, so that $|B_{i}|=|X|$. Then,
$\inputarr$ is the concatenation of $B_{1}\cdot C \cdot \ldots\cdot
C\cdot B_{n}\cdot C^{4} \cdot (B_{n})^{r}\cdot C \cdot\ldots\cdot C
\cdot (B_{1})^{r}$, where $C=\beta^{|X|}$ (and every $\beta$ is a
dummy symbol).  To check if $S_{i}\cap S_{j}=\emptyset$, we translate
it into a range $[i,j]$ starting at the first character encoding an
element of $S_{i}$ in $B_{i}$ and ending at the last character
encoding an element of $S_{j}$ in $(B_{j})^{r}$. The length of the
range is $|S_{i}|+|S_{j}|+|X|\cdot(2t+2)$, where $t=n-i+j-1$. We claim
that the range contains a $1/(2|X|)$-majority element iff $S_{i}\cap
S_{j}\neq \emptyset$. This is because such an element must occur
$\lceil t+1+(|S_{i}|+|S_{j}|)/(2|X|)\rceil=t+2$ times.  However, the
encoding of every set is a permutation of $X$, so every element $x\in
X$ occurs $t$ times in total in the middle part between $B_{i}$ and
$(B_{j})^{r}$.  Then, there are two additional occurrences of $x$
exactly when $x\in B_{i}$ and $x\in B_{j}$.

The strong version of the above conjecture implies that, for a string of length $2n\lg^{c}n$,
any structure $1/\lg^{c}n$-majority decision queries either needs $\Omega(n^{2})$
space or takes $\Omega(\lg^{c-1}n)$ time to answer a query.

\section{\label{sec:handling-small-tau}Missing Details for Theorem~\ref{thm:upper-bound}} 

\subsection{Preliminaries: Sequences on Larger Alphabets}

In addition to the bit vectors operations we defined earlier, also make
use of generalized wavelet trees, which generalize rank and select
operations to larger alphabets:

\begin{lemma}[\cite{GGV03,FMMN07}]
\label{lem:wavelet}
Given an array $S[1..n]$ with elements drawn from the range
$[1,\sigma]$ we can store $S$ using $n \lg \sigma + o(n \lg \sigma)$
bits of space, such that the following operations can be supported in
$\Oh(1 + \lg \sigma / (\lg \lg n))$ time:
\begin{enumerate}
\item $\access(S,i)$ return the element $S[i]$.
\item $\rank_\alpha(S,i)$: return the number of occurrences of symbol $
\alpha$ in $S[1..i]$.
\item $\select_\alpha(S,j)$: return the position of the $j$-th occurrence of $\alpha$ in $S$, if it exists, and $-1$ otherwise.
\end{enumerate}
\end{lemma}

\noindent Furthermore, given a position $i \in [1,n]$ a wavelet tree
can support a \emph{batch rank} operation in $\Oh(\sigma)$ time that
returns $\rank_\alpha(S,i)$ for all $\alpha \in
[1,\sigma]$~\cite{GHMN11,DHMNS13}.

\subsection{Handling small values of $\tau$}

Consider quadruple $\quadruple$ at level $k$, which has $\ell =
\Oh(1/\tau)$ active candidates, and let $[i',j']$ be the range spanned
by the extent of $\quadruple$.  Note that we can compute $i'$ and $j'$
in $\Oh(1)$ time using $L_k$ and Lemma~\ref{lem:find-level}.  Let
$M(\quadruple)$ be a bit vector of length $j' - i' + 1$, in which we
put a one at position $k$ if $A[i' + k]$ was an active candidate, and
$0$ otherwise.  Suppose we store $M(\quadruple)$ using
Lemma~\ref{lem:rank-select} and the same concatenation trick as
before: the leading term in the space bound will be no more than the
previous approach, but now the redundancy becomes $\Oh(n /
\textup{polylog}(n))$ for each level, and does not depend on $\tau$.
The problem is that we have lost the ability to distinguish between
the different candidates for each quadruple.  To do this, we need to
define some additional structures, and make use of the following
technical lemma:

\begin{lemma}
\label{lem:wt-partition}
A sequence $S[1..n]$ of elements drawn from the range $[1,\sigma]$ can
be stored using $\Oh(n \lg \sigma)$ bits, so that given subset of the
elements $\mathcal{Y} \subseteq [1,\sigma]$ and a range $[i,j]$ the frequency of
every $y \in \mathcal{Y}$ in $S[i..j]$ can be computed in time
$\Oh(|\mathcal{Y}| + \sigma / \lg n)$ time.
\end{lemma}

\begin{proof}
We partition the elements in $S$ into \emph{groups} of size $\lceil
\lg n \rceil$, so that symbol $i$ is in group $g(i) = \lceil i/\lceil
\lg n \rceil \rceil$.  Let $S'$ be the sequence such that $S'[i] =
g(S[i])$.  Finally, for each group $z \in [1,g(\sigma)]$, let $S_z[i]
= S[\select_z(S',i)]$.  We construct and store the sequences $S'$, and
$S_z$ for $z \in [1,g(\sigma)]$ using Lemma~\ref{lem:wavelet}.  In
particular we use two wavelet trees: one for $S'$ and one for the
concatenation $S_1,\ldots, S_{g(\sigma)}$.  Since the alphabet size of
$S'$ is $g(\sigma) = \Oh(\sigma /\lg n)$ and the alphabet size of each
$S_z$ is $\Oh(\lg n)$, the total space is $\Oh(n \lg (\sigma / \lg n)
+ n \lg \lg n) = \Oh(n \lg \sigma)$ (we assume that $\sigma \geq \lg n$
as otherwise we can query the wavelet tree with every $u\in \mathcal{Y}$
separately).

For the query $[i,j]$, we make use of the batch rank query operation on
$i$ and $j$ to remap the query to the appropriate range for each $S_z$
individually.  We also perform the batch rank query on position $n$ in
$S'$ to compute the lengths of each $S_z$.  Using these lengths we can
compute the start and end positions of $S_z$ for each $z \in
[1,g(\sigma)]$.  These batch queries take $\Oh(g(\sigma)) =
\Oh(\sigma/\lg n)$ time.  For each $y \in \mathcal{Y}$, we can compute
their group $z = g(y)$ and their offset in the second wavelet tree
using the partial sums and results of the batch rank queries.
Finally, in constant time (since $S_z$ has $\Oh(\lg n)$ distinct
elements), we can compute rank queries on corresponding to the range
$[i,j]$ in $S_z$ to get the frequency of each element in $\mathcal{Y}$
that happens to be in group $z$.  Since each individual element takes
$\Oh(1)$ time to process after the initial batch query on $S'$, we use
$\Oh(|\mathcal{Y}| + \sigma/\lg(n))$ time in total.
\qed
\end{proof}

Consider the subsequence $P$ of $A[i'..j']$ induced by the one bits in
$M$: i.e., $P[i] = A[i' + \select(M(\quadruple),i)]$.  Moreover,
suppose we replace the elements in $P$ by their ranks in the implicit
ordering of candidates in $\quadruple$.  Thus, $P$ has an alphabet
from the range $[1, \ldots, \sigma_P = \Oh(1/\tau)]$, and we represent
it using the data structure of Lemma~\ref{lem:wt-partition}.  At each
level, we concatenate the structure for $P$ for each quadruple, and
for each quadruple this structure costs
$\Oh(\rank(M(\quadruple),j'-i'+1)\lg(1/\tau))$ bits.  Thus, by the same
charging argument as in Lemma~\ref{lem:encoding-structure}, we can
bound the total cost of these concatenated sequences $P$ on all
levels by $\Oh(n \lg (1/\tau))$ bits.

Now, consider a query $[i,j]$.  The quadruple $\quadruple$ associated
with $[i,j]$ has a bit vector $M(\quadruple)$.  Using $M(\quadruple)$
we remap the query $[i,j]$ to $[\ell_1 = \rank(M(\quadruple),i - i' +
  1), \ell_2 = \rank(M(\quadruple),j - i' + 1)]$ in $\Oh(1)$ time.
Our goal is:
\begin{enumerate}
\item Extract the frequencies of all candidates that were active in $\quadruple$.
\item Follow up to $\Oh(\lg n)$ pointers to search for candidates that
  were inactive in $\quadruple$.  Suppose for one such pointer there
  are $q$ candidates to verify.
\end{enumerate}

Since the number of distinct elements in each $P$ is $\Oh(1/\tau)$, we
have that at each level we can verify $q$ candidates in $\Oh(q+1/(\tau \cdot \lg n))$
time.  Thus, we can verify all candidates
associated with $\quadruple$ in time $\Oh(1/\tau)$, as there are at
most $\lg n$ levels. This completes the proof of
Theorem~\ref{thm:upper-bound}.

\section{\label{sec:figures}Additional Figures}

\begin{figure}[h]
\includegraphics[width=\textwidth]{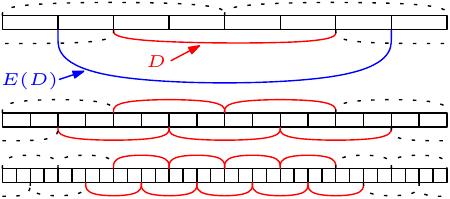}
\caption{\label{fig:extents}A quadruple $\quadruple$ and its extent
  $E(\quadruple)$.  Quadruples that overlap $\quadruple$ at lower
  levels (shown in red) are fully contained in the extent
  $E(\quadruple)$.}
\end{figure}

\begin{figure}[h]
\includegraphics[width=\textwidth]{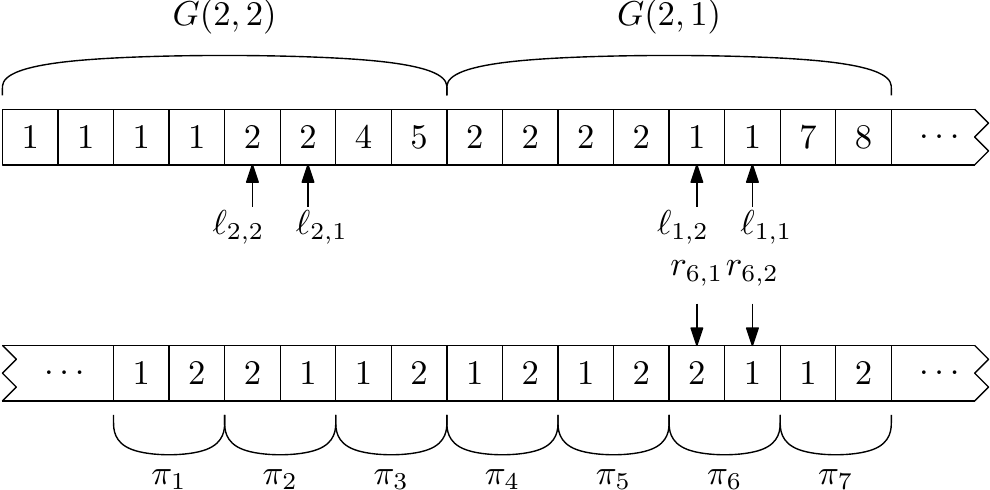}
\caption{\label{fig:lower-bound}Example of $\inputarr$ for the case where
  $k=2$, as well as other notational definitions.}
\end{figure}

\end{document}